\documentclass[journal]{IEEEtran}
\usepackage{graphicx,amssymb,lineno}
\usepackage{amsmath,amsfonts,amssymb,amsthm}
\newtheorem{theorem}{Theorem}

\newtheorem{definition}{Definition}
\usepackage{algorithm}
\usepackage{algorithmic}
\usepackage[usenames]{color}
\usepackage{float}
\usepackage{mathtools}
\usepackage{cite}

\usepackage{graphicx,graphics,color,epsfig,subfigure,graphpap,rotate}
\usepackage{times, verbatim, subfigure, epsfig, graphicx, latexsym, amsmath}
\usepackage{url}
\usepackage{subfigure}

\begin{document}

\title{Resource Allocation for Device-to-Device Communications Underlaying Heterogeneous Cellular Networks Using Coalitional Games}

\author{Yali~Chen,
        Bo~Ai,~\IEEEmembership{Senior Member,~IEEE},
        Yong~Niu,~\IEEEmembership{Member,~IEEE},
        Ke Guan,
        and Zhu~Han,~\IEEEmembership{Fellow,~IEEE}
\thanks{Y. Chen, B. Ai, Y. Niu, and K. Guan are with the State Key Laboratory of Rail Traffic Control and Safety, Beijing Engineering Research Center of High-speed Railway Broadband Mobile Communications, and the School of Electronic and Information Engineering, Beijing Jiaotong University, Beijing 100044, China (e-mails:
niuy11@163.com; boai@bjtu.edu.cn).}

\thanks{Z. Han is with the University of Houston, Houston, TX 77004 USA (e-mail: zhan2@uh.edu), and also with the Department of Computer Science and Engineering, Kyung Hee University, Seoul, South Korea.}

\thanks{This study was supported by the National key research and development program under Grant 2016YFB1200102-04 and 2016YFE0200900; and by the Fundamental Research Funds for the Central Universities Grant 2016RC056; and by NSFC under Grant 61725101 and Beijing natural fund under Grant L172020; and by the State Key Lab of Rail Traffic Control and Safety under Grant 2017JBM332, RCS2017ZZ005; and by the State Key Lab of Rail Traffic Control and Safety (Contract No. RCS2017ZT009), Beijing Jiaotong University; and by Teaching reform project under Grant 134493522 and 134496522; and by the China Postdoctoral Science Foundation under Grant 2017M610040; and by the National S\&T Major Project 2016ZX03001021-0033; and by the Natural Science Foundation of China (Grant No. : U1334202), and partially by US NSF CNS-1717454, CNS-1731424, CNS-1702850, CNS-1646607, ECCS-1547201. (\emph{Corresponding authors: B. Ai, Y. Niu.})}
}

\maketitle

\begin{abstract}

Heterogeneous cellular networks (HCNs) with millimeter wave (mmWave) communications included are emerging as a promising candidate for the fifth generation mobile network. With highly directional antenna arrays, mmWave links are able to provide several-Gbps transmission rate. However, mmWave links are easily blocked without line of sight. On the other hand, D2D communications have been proposed to support many content based applications, and need to share resources with users in HCNs to improve spectral reuse and enhance system capacity. Consequently, an efficient resource allocation scheme for D2D pairs among both mmWave and the cellular carrier band is needed. In this paper, we first formulate the problem of the resource allocation among mmWave and the cellular band for multiple D2D pairs from the view point of game theory. Then, with the characteristics of cellular and mmWave communications considered, we propose a coalition formation game to maximize the system sum rate in statistical average sense. We also theoretically prove that our proposed game converges to a Nash-stable equilibrium and further reaches the near-optimal solution with fast convergence rate. Through extensive simulations under various system parameters, we demonstrate the superior performance of our scheme in terms of the system sum rate compared with several other practical schemes.

\end{abstract}


\section{Introduction}\label{S1}

With the increasing proliferation of mobile devices with high capabilities and intelligence, the global mobile traffic is expected to experience a remarkable and continuous growth in the next few years. As predicted by Cisco, the traffic generated from wireless and mobile devices is expected to constitute a major percentage of the total internet protocol (IP) traffic by 2020. It is also estimated that the number of devices accessed to IP networks will be three times of the global population in 2020 and the mobile traffic will grow at an annual rate of $53\%$ until 2020 \cite{Cisco}. At the same time, the millimeter wave (mmWave) has huge bandwidth, and therefore, much higher network capacity can be achieved \cite{Future}. There are already several standards defined for indoor wireless personal area networks (WPANs) or wireless local area networks (WLANs) in the mmWave band, such as ECMA-387 \cite{ECMA387}, IEEE 802.15.3c \cite{IEEE802153c}, and IEEE 802.11ad.
Thus, in order to keep up with the explosive growth of mobile devices and data traffic, one key enabling solution is to exploit HCNs in both the cellular band and the mmWave band.

HCNs operating in both conventional cellular band and in the mmWave band, can improve the system performance effectively. Two kinds of networks offer different advantages. For example, cellular network provides higher link reliability, while mmWave communication has obvious advantages in the transmission rate. However, a most common concern is that mmWave communications suffer a much larger distance-dependent propagation loss due to the high carrier frequency \cite{niu2017energy,ai2017indoor,YongNiu}. For example, the free space path loss at the 60 GHz band is 28 dB more than that at 2.4 GHz \cite{singh2011interference}. To combat severe channel attenuation, we utilize the highly directional antennas and the beamforming technology at both the transmitter and the receiver \cite{roh2014millimeter}. Moreover, mmWave communication typically requires line of sight (LOS) communication.

D2D communications underlaying the HCN, as a method of great potential to offload traffic from the base station (BS), can improve network performance and provide a better user experience \cite{tan2017joint,song2013wireless}. Under the coverage of BS, user equipments (UEs) in physical proximity communicate with each other directly using the resources in the mmWave band or sharing resources with cellular users. The integration of D2D communications into HCNs has the advantage of allowing for high data rate, low delay and power consumption transmission for popular proximity-based applications \cite{Energy}. Consequently, these high quality D2D links generate the hop gain by transmitting data signals directly between two closely located terminals without involving a centralized controller. On the other hand, reuse gain is achieved by simultaneously using the same radio resource for cellular users and D2D pairs. Additionally, the D2D-enabled HCN also facilitates new types of peer-to-peer services.

\begin{figure}[htbp]
\begin{center}
\includegraphics*[width=0.9\columnwidth,height=2.5in]{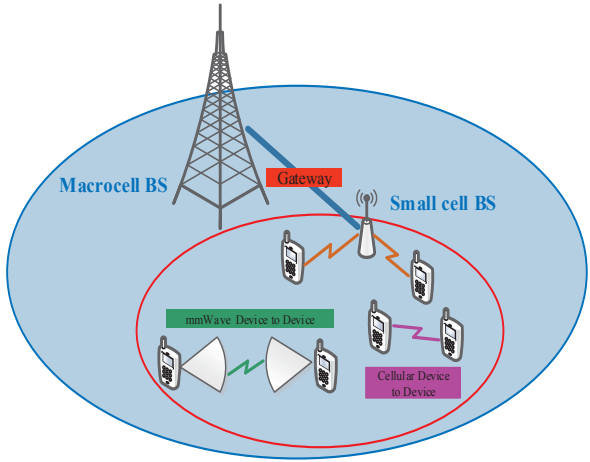}
\end{center}
\caption{The D2D-enabled HCN underlaying the macrocell.} \label{fig1}
\end{figure}

In Fig. \ref{fig1}, we show a typical scenario of the D2D-enabled HCN underlaying the macrocell. Cellular users are associated with BS of the small cell, which is connected to the BS of the macrocell via the gateway. In the D2D-enabled HCN, interference produced by D2D communications hampers the performance of cellular communications. Intra-cell interference, which is referred to the interference between users as the result of spectrum sharing, is considered to be an important and complex problem in HCNs, especially the interference between D2D pairs \cite{ma2015interference}. Therefore, it is necessary to investigate and properly deal with the interference problems such that the benefits of proximity transmissions can be fully exploited. To date, extensive works have been undertaken on the power control \cite{ramezani2017joint,kaufman2008cellular,yu2014joint}, resource allocation \cite{xu2013efficiency,wang2016social,wang2013dynamic,li2014coalitional,zhao2015social} and association techniques among the cellular users and D2D pairs to mitigate the interference and obtain the maximum system achievable transmission rate. Besides, based on the differences between cellular D2D networks and mmWave D2D networks, how to utilize the advantages of both networks to optimize the sub-channel allocation under HCNs indeed brings great challenges.

In this paper, we consider D2D communications in the HCN combining mmWave and cellular networks for uplink resource allocation, and then formulate the problem of maximizing the system sum rate via resource allocation into a nonlinear integer programming problem. With the complicated interferences considered among cellular users and D2D pairs, we address the problem of resource allocation for multiple cellular users and D2D pairs from a game theory point of view using coalition formation game \cite{saad2011coalitional}. The coalition game, which is widely used in wireless communications, for example, the resource allocation problems, allows several players cooperatively to form a coalition in order to optimize resource allocation, manage the interference, and further enhance the system performance. Then, we develop a coalition formation algorithm to achieve the Nash-stable equilibrium for the proposed coalition game. The main contributions of this paper can be summarized as follows.

\begin{itemize}
\item We introduce the coalition formation game to model the D2D communications underlaying HCN consisting of multiple cellular users and D2D pairs. Based on the established model, we investigate the resource allocation problems for the realistic HCNs.

\item We formulate the problem of D2D resource allocation underlaying HCN aiming to enable massive connectivity and maximize the system sum rate. Then, we utilize the advantages of cellular D2D network and mmWave D2D network, and develop a coalition formation algorithm to implement efficient resource allocation with low computation complexity. We show that the proposed algorithm converges to a Nash-stable coalition structure and achieves a near-optimal solution with fast convergence rate.

\item Through extensive simulations under various system parameters, we evaluate the system performance of our proposed coalition game based approach compared with other practical schemes.

\end{itemize}

The rest of the paper is organized as follows. In Section~\ref{S2}, we present an overview of the related work. Section~\ref{S3} introduces the system model and formulates the resource allocation problem. The coalition game with transferable utility and corresponding algorithm is proposed in Section~\ref{S4}. We analyze the properties of the proposed algorithm in Section~\ref{S5}. Section~\ref{S6} gives the performance evaluation of our proposed scheme compared with other schemes under various system parameters. Finally, the conclusions of this paper are drawn in Section~\ref{S7}.

\section{Related Work}\label{S2}

There have been several related works studying resource allocation and interference management for D2D communications. For example, Ramezani-Kebrya \emph{et al.} \cite{ramezani2017joint} proposed an efficient power control algorithm and jointly optimized the power of a cellular user and a D2D pair aiming at maximizing their sum rate, while providing a lower bound on the signal to interference plus noise ratio (SINR) requirements. Kaufman \emph{et al.} \cite{kaufman2008cellular} proposed that D2D users determined their path loss to the BS according to the received power in the downlink, and then adjusted the transmit power so that the interference caused by D2D users to the BS is minimized. Yu \emph{et al.} \cite{yu2014joint} improved the system performance in terms of throughput by investigating power control, channel assignment and mode selection. Xu \emph{et al.} \cite{xu2013efficiency} proposed an innovative reverse iterative combinatorial auction mechanism to allocate resources to D2D communications underlaying downlink cellular networks. The above works have shown that involving D2D communications can improve the overall system performance by proper resource allocation and reasonable management of interference among cellular and D2D pairs. Compared with the related work, our paper aims to solve the problem of D2D resource allocation in HCNs, and there is no doubt that the interference problems are of great complexity. In this paper, we consider a scheme from the view point of game theory to maximize the system sum rate.

Game theory offers a set of mathematical tools to study the complex interactions among interdependent rational players and to predict their choices of strategies \cite{Wilson1992Game}. Besides, with many different game methods included, the game theory has attracted considerable attentions. The related researches utilizing the game theory in the field of wireless communication include the analysis of the resource allocation problems, especially the spectrum allocations in the cellular and heterogeneous networks. Wang \emph{et al.} \cite{wang2016social} studied the community-aware D2D resource allocation and further proposed a two-step coalition game to implement effective resource allocation underlaying cellular networks. Wang \emph{et al.} \cite{wang2013dynamic} proposed a cooperative coalition game to cope with the problem that on-board units might not have the ability to complete the download task of the entire large file from the roadside unit when moving at high speed in vehicular ad hoc networks. In order to improve spectrum efficiency, Li \emph{et al.} \cite{li2014coalitional} proposed a coalition formation game to address the problem of uplink resource
allocation for multiple cellular users and D2D pairs. Combining both the interference constraints in the physical domain and social connections in the social domain, Zhao \emph{et al.} \cite{zhao2015social} proposed a social group utility maximization game based D2D resource allocation scheme to maximize each D2D user's social group utility. However, the coalition game in related work aims to find a coalitional structure that maximizes the individual payoffs of the players, while we entail finding a structure that maximizes the total utility.

MmWave communication is considered to be one of the most concerned candidate technologies for the fifth generation (5G). The fact that lower frequencies of the radio spectrum have become saturated and are unable to meet the exponential growth in traffic demand, has motivated the exploration of the under-utilized mmWave frequency spectrum for future high-speed broadband cellular networks \cite{Zhenyu,Zhenyu2,Guan2017}. However, mmWave communications have unique characteristics that are different from traditional cellular networks. On the one hand, mmWave communication is typically characterized by transmission and reception with very narrow beams and highly directional antenna. On the other hand, mmWave communication suffers a much larger propagation loss due to the high carrier frequency, and mmWave links are easily blocked by human body and other obstacles. Consequently, network congestion may happen in mmWave networks \cite{BlockageRobust}. There are some works on utilizing mmWave band in wireless network. Ai \emph{et al.} \cite{ai2017indoor} performed some measurements and simulations on indoor mmWave massive multiple-input multiple-output (MIMO) channel at a band in 26 GHz. Shariat \emph{et al.} \cite{shariat2016radio} presented some important findings in designing radio resource management (RRM) functionalities of mmWave in conjunction with heterogeneous network in both backhaul and access links. Rebato \emph{et al.} \cite{rebato2017hybrid} proposed an effective novel hybrid spectrum access scheme consisted of the exclusive low frequency carrier and the pooled high frequency carrier for mmWave networks. Niu \emph{et al.} \cite{niu2017energy} developed an energy-efficient mmWave backhauling scheme to deal with the joint optimization problem of concurrent transmission scheduling and power control of small cells densely deployed in HCNs.

\section{System Overview and Problem Formulation}\label{S3}

\begin{figure}[htbp]
\begin{center}
\includegraphics*[width=0.9\columnwidth,height=2.5in]{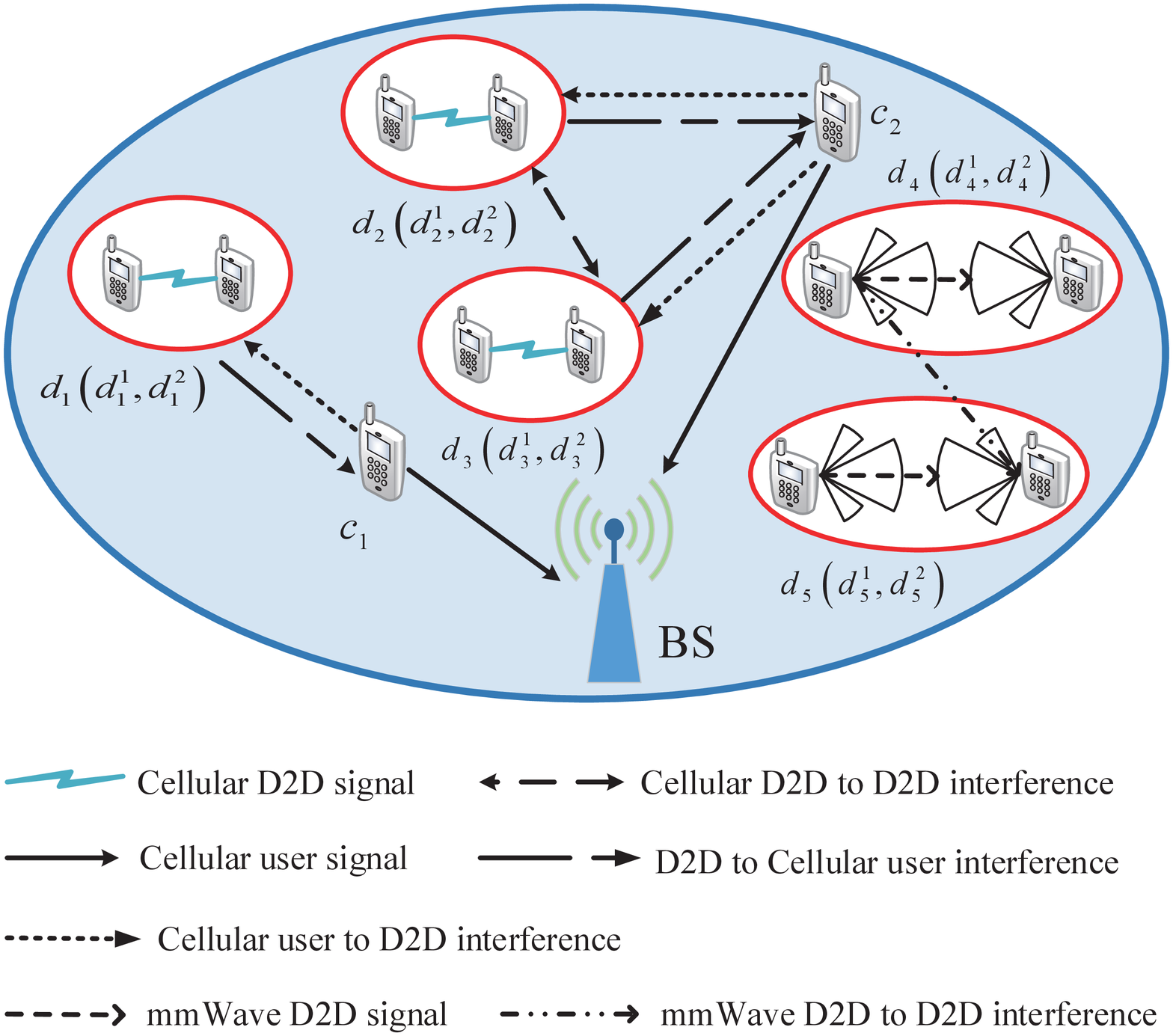}
\end{center}
\caption{Illustration of the resource sharing of D2D communications underlaying HCN, where there are 2 cellular users, $c_{1}$ and $c_{2}$, and 5 D2D pairs.} \label{fig2}
\end{figure}

In this section, we first give a system overview for D2D communications underlaying HCN, and then formulate the resource allocation problem by defining optimization utility function that reflects the system performance in terms of system sum rate.

\subsection{System Description}\label{S3-1}
We consider a scenario of a single cell coupled with all the users under its coverage. In our investigated system, we focus on the intra-cell interference generated by the users sharing the same frequency band. Since the heterogeneous network consists of the cellular band and mmWave band, there are two kinds of modes to select for each D2D pair. One is to share the uplink spectrum resource of one cellular user, and the other is to use the resource in mmWave band. On the one hand, we consider the cellular D2D network, where the BS is equipped with omnidirectional antennas for cellular communications. We assume that the cellular users share their uplink resources with D2D communications when the cellular access mode is selected by D2D pairs, and one cellular user's spectrum resource can be shared with multiple D2D pairs to achieve the maximum spectral efficiency, while we also assume that a D2D pair shares no more than one cellular user's uplink resource for the purpose of reducing interference caused by D2D communications and decreasing the corresponding complexity. In addition, it is supposed that the subcarrier channels occupied by cellular users are mutually independent for analytical tractability. In other words, D2D pairs will not interfere with each other when sharing different cellular users' uplink spectrum resources in cellular D2D network. On the other hand, we consider the mmWave D2D network, which doesn't require infrastructure such as BSs. Millimeter wave communication is equipped with the highly directional antenna in order to achieve the directional transmission and reception between D2D users in mmWave band \cite{niu2017energy}. With highly directional antenna arrays in mmWave, D2D pairs are able to share the same radio resource. As illustrated in Fig. \ref{fig2}, there exists two cellular users $c_{1}$ and $c_{2}$, and the D2D pair $(d^{1}_{1},d^{2}_{1})$ occupies the spectrum resource of $c_{1}$, while D2D pairs $(d^{1}_{2},d^{2}_{2})$ and $(d^{1}_{3},d^{2}_{3})$ occupy the spectrum resource of $c_{2}$. Besides, D2D pairs $(d^{1}_{4},d^{2}_{4})$ and $(d^{1}_{5},d^{2}_{5})$ use the spectrum resource in the mmWave band. On the whole, we only need to focus on the analysis of the signal interference between D2D pairs in mmWave band and the signal interference among cellular users and D2D pairs in cellular band.

In such a system, we concentrate on assigning appropriate uplink spectrum resources occupied by the cellular users or mmWave radio resource to D2D pairs in order to enhance the whole network performance. Since the D2D pair shares the same spectrum resources with the cellular users or with other D2D pairs in mmWave band, as the result of that, the system performance will be reduced to compensate the interference. In order to maximize the system performance, what we should do is to properly manage the interference and limit the interference as much as possible. As shown in Fig. \ref{fig2}, there are three kinds of interference in cellular D2D network, such as cellular D2D to D2D interference, D2D to cellular user interference and cellular user to D2D interference. The cellular user and its corresponding D2D pairs interfere with each other because they share the same uplink spectrum resources. The received signals at the BS from the cellular user $c$ are interfered by the transmitters of D2D pairs sharing the same spectrum resource of $c$. The signal at the D2D receiver $d$ is interfered by the cellular user $c$ and other D2D links sharing the same spectrum resource of $c$. On the other hand, there exists just one kind of interference in mmWave D2D network and the D2D pairs are mutually interfered as they use the same spectrum resource in mmWave band.

\subsection{System Model}\label{S3-2}

In the system, we assume there are $C$ cellular users labeled as the set of $\textbf{C}=\{c_{1},c_{2},...c_{C}\}$ that share their uplink resources with D2D pairs. Moreover, we denote the set of $D$ pairs of D2D users by $\textbf{D}$, written as $\textbf{D}=\{d_{1},d_{2},...d_{D}\}$. Every D2D pair independently randomly chooses to share the resource of any cellular user $c_{i},\forall{c_{i}}\in{\textbf{C}}$ or the resource in mmWave band. To better reflect the spectrum resource usage relationship, we define a binary variable $a_{d}$ for each D2D pair $d$ to represent whether the cellular or mmWave frequency band is selected. If the cellular frequency band is selected, $a_{d}=1$; otherwise, $a_{d}=0$. Besides, we define another binary variable $x_{c,d}$ to indicate whether the uplink spectrum resource of cellular user $c$ is shared by $d$, $\forall{c}\in{\textbf{C}}, {d}\in{\textbf{D}}$, where if $x_{c,d}=1$, it means that the resource blocks of cellular user $c$ are allocated to the D2D pair $d$, otherwise, $x_{c,d}=0$. We analyze the constraints of $x_{c,d}$. First, each D2D pair can share the uplink spectrum resource from no more than one cellular user, which can be expressed as $\sum\limits_{{c}\in{\textbf{C}}}x_{c,d}\leq{1}, \forall{d}\in{\textbf{D}}$. Second, $x_{c,d}$ is equal to $a_{d}$ for all D2D pairs, which can be expressed as $\sum\limits_{{c}\in{\textbf{C}}}x_{c,d}=a_{d}$. On the one hand, sharing the spectrum resource of one cellular user by multiple D2D pairs is allowed in this sharing model in order to increase the spectrum resource reuse ratio. On the other hand, it is also possible for D2D pairs to occur on the same part of the mmWave spectrum resource.

To maximize the network performance in terms of system sum rate, we should consider the key part of SINR. Assuming that in the cellular D2D network, we adopt the channel model of Rayleigh for small-scale fading with the propagation loss factor $n$, under which the instantaneous channel taps are the function of time and spatial locations \cite{luo2010simulation}. The power or second-order statistic of the channel, denoted by ${|h_{0}|}^2$, is a constant within the BS's coverage area. For communication link $i$, we denote its sender and receiver by $s_{i}$ and $r_{i}$, respectively. According to the path loss model, we derive the expression of the received power at $r_{i}$ from $s_{i}$ as $P^{c}_{r}(i,i)={|h_{0}|}^2\cdot{G_{t}}\cdot{G_{r}}\cdot{l}^{-n}_{ii}\cdot P_{c}$, where $P_{c}$ is the cellular transmission power, $l_{ii}$ is the distance between $s_{i}$ and $r_{i}$, $n$ is the path-loss exponent, $h_{0}$ is a complex Gaussian random variable with unit variance and zero mean, $G_{t}$ is the transmit antenna gain and $G_{r}$ is the receive antenna gain. Both of them are constants. The received SINR at $r_{i}$ from $s_{i}$ can be expressed as
\begin{equation}
SINR^{c}_{i}=\frac{{|h_{0}|}^2G_{t}G_{r}{l}^{-n}_{ii}{P_{c}}}{P^{c}_{int,i}+N_{0c}W_{c}}, \label{eq1}
\end{equation}
where $P^{c}_{int,i}$ is the interference signal power received by user $r_{i}$, $N_{0c}$ is the cellular onesided power spectral density of white Gaussian noise, and $W_{c}$ is the cellular subcarrier bandwidth.

Similarly, we assume that in the mmWave D2D network, the received power at $r_{i}$ from $s_{i}$ can be calculated as
\begin{equation}
P^{m}_{r}(i,i)=k_{0}G_{t}(i,i)G_{r}(i,i){l}^{-n}_{ii}P_{m}. \label{eq2}
\end{equation}
For two mutually independent communication links $i$ and $j$, the received interference at $r_{i}$ from $s_{j}$ can be calculated as
\begin{equation}
P^{m}_{r}(j,i)={\rho}k_{0}G_{t}(j,i)G_{r}(j,i){l}^{-n}_{ji}P_{m}, \label{eq3}
\end{equation}
where $k_{0}$ is a constant coefficient and proportional to ${(\frac{\lambda}{4\pi})}^2$ ($\lambda$ denotes the wavelength), $\rho$ denotes the multi-user interference (MUI) factor related to the cross correlation of signals from different links, and $P_{m}$ is the transmitted power of mmWave \cite{qiao2012stdma}. Unlike the assumption in cellular D2D network, the antenna gain of $s_{i}$ in the direction of $s_{i}\rightarrow r_{i}$ is denoted by $G_{t}(i,i)$ and is no longer a constant. The antenna gain of $r_{i}$ in the direction of $s_{i}\rightarrow r_{i}$ is denoted by $G_{r}(i,i)$. Thus, the received SINR at $r_{i}$ can be expressed as
\begin{equation}
SINR^{m}_{i}=\frac{P^{m}_{r}(i,i)}{{P^{m}_{int,i}}+N_{0m}W_{m}}, \label{eq4}
\end{equation}
where $P^{m}_{int,i}$ is the interference signal power received by user $r_{i}$, $N_{0m}$ is the mmWave onesided power spectral density of white Gaussian noise, and $W_{m}$ is the bandwidth of mmWave communication.

In the case of cellular communication, we abbreviate the transmit and receive antenna gain of device and BS as $G_{0}$ and $G_{b}$, respectively, since they are taken the fixed value in cellular D2D network. Then, we are able to obtain the uplink transmission rate corresponding to cellular users and D2D pairs. The BS receiving signal from the cellular user subjects to interference from D2D pairs referred to that occupying the same spectrum resource with cellular user. Therefore, the interference power at the BS for cellular user $c$ can be expressed as
\begin{equation}
P_{int,c}=\sum_{{d}\in{\textbf{D}}}x_{c,d}{|h_{0}|}^2{G_{0}}G_{b}{l}^{-n}_{db}P_{c}. \label{eq5}
\end{equation}
According to Shannon's channel capacity, the uplink channel rate of the cellular user $c$, denoted by $R_{c}$, is
\begin{equation}
R_{c}=W_{c}{\log}_2\left({1+\frac{{|h_{0}|}^2{G_{0}}G_{b}{l}^{-n}_{cb}P_{c}}{\sum\limits_{{d}\in{\textbf{D}}}x_{c,d}{|h_{0}|}^2{G_{0}}G_{b}{l}^{-n}_{db}P_{c}+N_{0c}W_{c}}}\right). \label{eq6}
\end{equation}
The D2D receiver $d$ suffers interference from the cellular user $c$ and the other D2D pairs sharing the same spectrum resource of $c$. Therefore, we can get the following expression of interference power for D2D receiver $d$, denoted by ${P}^{c}_{int,d}$.
\begin{equation}
\begin{aligned}
{P}^{c}_{int,d}&=\sum_{{c}\in{\textbf{C}}}x_{c,d}{|h_{0}|}^2{G_{0}}^{2}{l}^{-n}_{cd}P_{c}\\
&+\sum_{{{d}^{\prime}}\in{\textbf{D}\setminus\{d\}}}\sum_{{c}\in{\textbf{C}}}x_{c,d}x_{c,{d}^{\prime}}{|h_{0}|}^2{G_{0}}^{2}{l}^{-n}_{{d}^{\prime}d}P_{c}. \label{eq7}
\end{aligned}
\end{equation}
According to (\ref{eq7}), we can obtain the received SINR at the D2D receiver $d$, denoted by $SINR^{c}_{d}$, as follows.
\begin{equation}
SINR^{c}_{d}=\frac{{|h_{0}|}^2{G_{0}}^{2}{l}^{-n}_{dd}P_{c}}{{P}^{c}_{int,d}+N_{0c}W_{c}}. \label{eq8}
\end{equation}

In the case of mmWave communication, we can derive the transmission rate of D2D pairs similarly. The interference of D2D receiver $d$ is from the other D2D pairs in mmWave band. Thus, we can obtain the interference power from the other D2D pairs for D2D receiver $d$, denoted by ${P}^{m}_{int,d}$, as follows.
\begin{equation}
{P}^{m}_{int,d}=\sum_{{{d}^{\prime}}\in{\textbf{D}\setminus\{d\}}}(1-a_{{d}^{\prime}}){\rho}k_{0}G_{t}({d}^{\prime},d)G_{r}({d}^{\prime},d){l}^{-n}_{{d}^{\prime}d}P_{m}.
\label{eq9}
\end{equation}
According to (\ref{eq9}), we can get the following received SINR at the D2D receiver $d$, denoted by $SINR^{m}_{d}$.
\begin{equation}
SINR^{m}_{d}=\frac{k_{0}G_{t}(d,d)G_{r}(d,d){l}^{-n}_{dd}P_{m}}{{P}^{m}_{int,d}+N_{0m}W_{m}}. \label{eq10}
\end{equation}

Combining the $SINR^{c}_{d}$ in cellular D2D network and the $SINR^{m}_{d}$ in mmWave D2D network, the SINR received by D2D receiver $d$ in HCN, denoted by $SINR_{d}$, can be calculated as
\begin{equation}
SINR_{d}=a_{d}SINR^{c}_{d}+(1-a_{d})SINR^{m}_{d}. \label{eq11}
\end{equation}
The achievable channel rate for the D2D pair $d$, denoted by $R_{d}$, is give in (\ref{eq12}), shown at the top of the next page.
\newcounter{mytempeqncnt}
\begin{figure*}[!t]
\normalsize
\setcounter{mytempeqncnt}{\value{equation}}
\setcounter{equation}{11}
\begin{equation}
\begin{aligned}
R_{d}&=a_{d}W_{c}{\log}_2\left(1+SINR^{c}_{d}\right)+(1-a_{d})W_{m}{\log}_2\left(1+SINR^{m}_{d}\right)\\
&=a_{d}W_{c}{\log}_2\left(1+\frac{{|h_{0}|}^2{G_{0}}^{2}{l}^{-n}_{dd}P_{c}}{\sum\limits_{{c}\in{\textbf{C}}}x_{c,d}{|h_{0}|}^2{G_{0}}^{2}{l}^{-n}_{cd}P_{c}+\sum\limits_{{{d}^{\prime}}\in{\textbf{D}\setminus\{d\}}}\sum\limits_{{c}\in{\textbf{C}}}x_{c,d}x_{c,{d}^{\prime}}{|h_{0}|}^2{G_{0}}^{2}{l}^{-n}_{{d}^{\prime}d}P_{c}+N_{0c}W_{c}}\right)\\
&+(1-a_{d})W_{m}{\log}_2\left(1+\frac{k_{0}G_{t}(d,d)G_{r}(d,d){l}^{-n}_{dd}P_{m}}{\sum\limits_{{{d}^{\prime}}\in{\textbf{D}\setminus\{d\}}}(1-a_{{d}^{\prime}}){\rho}k_{0}G_{t}({d}^{\prime},d)G_{r}({d}^{\prime},d){l}^{-n}_{{d}^{\prime}d}P_{m}+N_{0m}W_{m}}\right). \label{eq12}
\end{aligned}
\end{equation}
\setcounter{equation}{\value{mytempeqncnt}}
\hrulefill
\vspace*{4pt}
\end{figure*}
\setcounter{equation}{12}

Thus, the achieved system sum rate considering all the cellular users and D2D pairs in HCN, denoted by $R$, can be obtained as
\begin{equation}
R=\sum_{c\in{\textbf{C}}}R_{c}+\sum_{d\in{\textbf{D}}}(a_{d}R_{d}+(1-a_{d})(1-P_{out:d,d})R_{d}), \label{eq13}
\end{equation}
where $P_{out:d,d}$ denotes the probability of blockage in the LOS path between the sender and the receiver of D2D pair $d$ in mmWave band. It can be expressed as $P_{out:i,j}=1-e^{-\beta{l}_{ij}}$, where $l_{ij}$ is the distance between users $i$ and $j$, and $\beta$ is the parameter used to reflect the density and size of obstacles, which result in an interruption caused by blockage \cite{Lee2016Connectivity}.

\subsection{Problem Formulation}\label{S3-3}
Obviously, the system sum rate is related to the resource sharing relations $x_{c,d}$ and $a_{d}, \forall{c}\in{\textbf{C}}, {d}\in{\textbf{D}}$. In view of the relationship between these two binary variables, $\sum\limits_{{c}\in{\textbf{C}}}x_{c,d}=a_{d}$, $\forall{d}\in{\textbf{D}}$, we can define a system utility function that reflects the network performance as the system sum rate, denoted by $R({\mathbf X})$, where ${\mathbf X}$ is the matrix of $x_{c,d}, \forall{c}\in{\textbf{C}}, {d}\in{\textbf{D}}$. Therefore, based on the above analysis, the problem of determining the optimal resource allocation strategy in the D2D communications underlaying HCN to maximize the system sum rate can be formulated as follows.
\begin{align}
&\max\,\,R({\mathbf X})  \notag \\
&s.t.\quad
\begin{cases}
x_{c,d}\in\{0,1\},  \  \forall{d}\in{\textbf{D}}, {c}\in{\textbf{C}};\\
\sum\limits_{{c}\in{\textbf{C}}}x_{c,d}\leq {1}, \  \forall{d}\in{\textbf{D}}. \label{eq14}
\end{cases}
\end{align}

This is a nonlinear integer programming problem, where $x_{c,d}$ is the integer binary variable. In the formulated problem, the optimization utility function in (\ref{eq14}) has no obvious increasing or concave properties with $x_{c,d}$ even the constraint is linear. Obviously, this problem is NP-complete and it is more complex compared with the 0-1 Knapsack problem \cite{pisinger2005hard}. Our optimization problem aims to maximize the system sum rate. In the next section, we propose a coalition formation algorithm from the perspective of game theory to solve the problem with low complexity. For each D2D pair $d$ in the system, or equivalently each player in the game, it makes a decision on selecting the mmWave band or sharing the spectrum of the cellular user $c \ (c\in \textbf{C})$, only for making a greater contribution to the system utility function.

\section{Coalitional Game Approach}\label{S4}
In this section, we present the coalition game from the view point of game theory to solve the formulated resource sharing problem. Based on it, the coalition formation algorithm is proposed.

\subsection{Coalitional Game Formulation} \label{S4-1}
The formulated optimization problem aims to maximize the overall system performance. Based on the problem, we introduce a coalition game theory model, where the D2D pairs tend to form coalitions so that the system utility will improve. In our investigated system, there are $C$ cellular users and $D$ D2D pairs. The D2D pairs can choose to occupy the spectrum resource of any of the $C$ cellular users or use the resource in mmWave band. Thus, we suppose that there are $C+1$ coalitions formed by D2D pairs. We denote the coalitions as $F=\{F_{c_{1}},F_{c_{2}},...,F_{c_{C}},F_{c_{C+1}}\}$, where $F_{c_{x}} \bigcap F_{c_{{x}^{\prime}}}=\emptyset$ for any $x\neq x^{\prime}$, and $\bigcup_{x=1}^{C+1}F_{c_{x}}=\textbf{D}$. The cardinality of $F$ is the number of coalitions. We divide the coalitions into two groups for discussion. The first group is composed of coalitions of $F_{c}\subset F \ (c\in \textbf{C})$ sharing the resource with cellular user $c\in \textbf{C}$. The achieved uplink transmission rate of cellular user $c$ in this case can be written as
\begin{equation}
R_{c}=W_{c}{\log}_2\left({1+\frac{{|h_{0}|}^2{G_{0}}G_{b}{l}^{-n}_{cb}P_{c}}{\sum\limits_{d\in {F_{c}}}{|h_{0}|}^2{G_{0}}G_{b}{l}^{-n}_{db}P_{c}+N_{0c}W_{c}}}\right). \label{eq15}
\end{equation}
The uplink transmission rate of D2D pair $d \ (d\in F_{c})$ is given in (\ref{eq16}), shown at the top of the next page.

\begin{figure*}[!t]
\normalsize
\setcounter{mytempeqncnt}{\value{equation}}
\setcounter{equation}{15}
\begin{equation}
R_{d}=W_{c}{\log}_2\left(1+\frac{{|h_{0}|}^2{G_{0}}^{2}{l}^{-n}_{dd}P_{c}}{{|h_{0}|}^2{G_{0}}^{2}{l}^{-n}_{cd}P_{c}+\sum\limits_{{{d}^{\prime}}\in{F_{c}\setminus\{d\}}}{|h_{0}|}^2{G_{0}}^{2}{l}^{-n}_{{d}^{\prime}d}P_{c}+N_{0c}W_{c}}\right).
\label{eq16}
\end{equation}
\setcounter{equation}{\value{mytempeqncnt}}
\hrulefill
\vspace*{4pt}
\end{figure*}
\setcounter{equation}{16}

Consequently, the rate of the uplink channel shared by cellular user $c$ and D2D pairs $d\in F_{c}$, denoted by $R(F_{c})$, is given by

\begin{equation}
R(F_{c})=R_{c}+\sum\limits_{{d}\in{F_{c}}}R_{d}. \label{eq17}
\end{equation}

The other group is coalition $F_{c}\subset F \ (c=c_{C+1})$ sharing the resource in mmWave band. The channel rate of D2D pair $d \ (d\in F_{c})$ can be written as
\begin{equation}
\begin{aligned}
&R_{d}=\\
&W_{m}{\log}_2\left(1\!+\!\frac{k_{0}G_{t}(d,d)G_{r}(d,d){l}^{-n}_{dd}P_{m}}{\sum\limits_{\mathclap{{{{d}^{\prime}}\in{F_{c}\setminus\{d\}}}}}{\rho}k_{0}G_{t}({d}^{\prime},d)G_{r}({d}^{\prime},d){l}^{-n}_{{d}^{\prime}d}P_{m}\!+\!N_{0m}W_{m}}\right).
\label{eq18}
\end{aligned}
\end{equation}

Therefore, the rate of the channel occupied by D2D pairs $d\in F_{c}$, denoted by $R(F_{c})$, is given by
\begin{equation}
R(F_{c})=\sum\limits_{d\in F_{c}}(1-P_{out:d,d})R_{d}. \label{eq19}
\end{equation}

Obviously the larger the number of D2D pairs in a coalition, the greater the resulting interference among users. In the proposed coalitional game, if all the D2D pairs form a grand coalition to share one cellular user's uplink spectrum resource or the resource in mmWave band, no D2D pair can make a greater contribution to the system utility due to the severe interference. Therefore, all the D2D pairs are with little incentive to form a grand coalition. In addition, the mmWave communication rate is about six orders of magnitude larger than that of cellular communication. Thus, multiple D2D pairs will choose to share the resource in mmWave band, and some of the coalitions sharing the resources of cellular users may be empty for the purpose of maximizing the system sum rate. In this paper, the D2D resource allocation underlaying HCN is modeled in the coalitional game with transferable utilities, where the D2D pairs, as the game players, tend to form coalitions to share the resources of cellular users or mmWave radio resource in order to maximize the system sum profits. Finally, we define the proposed coalitional game with the transferable utility as follows.

\begin{definition}
Coalitional Game With Transferable Utility:
The concept of coalitional game with transferable utility has been first proposed by Morgenstern and von Neumann \cite{von2007theory}. A coalitional game with a transferable utility for D2D resource allocation underlaying HCN is defined by a pair ($\textbf{D},R$), where $\textbf{D}$ is the set of game players and $R$ is the payoff function. Both of them are the basic elements of game theory. $\forall \  F_{c}\subset F, \ R(F_{c})$ is a real number, which represents the sum profits contributed by the entire coalition $F_{c}$,  and it can be assigned to the members of coalition $F_{c}$ in any random way. Next, we define the coalition game for the proposed resource sharing relations.
\label{define1}
\end{definition}

\begin{definition}
Coalitional Game for D2D Resource Allocation:
The coalitional game with transferable utility for resource allocation of D2D communications is defined by the triple $(\textbf{D},R,F)$, where the set of the D2D pairs $\textbf{D}$ is players, $R$ is the transferable utility including the transmission rates of all the users in the coalition, and $F$ is the coalition partition, which can be denoted as $F=\{F_{c_{1}},F_{c_{2}},...,F_{c_{C}},F_{c_{C+1}}\}$, where $F_{c_{x}} \bigcap F_{c_{{x}^{\prime}}}=\emptyset$ for any $x\neq x^{\prime}$, and $\bigcup_{x=1}^{C+1}F_{c_{x}}=\textbf{D}$. It is a strategy for each D2D pair $d$ to make a decision on which coalition to share resources based on the system sum utility.
\label{define2}
\end{definition}

\subsection{Coalition Formation Algorithm}\label{S4-2}
In this subsection, we devise a coalition formation algorithm for the proposed coalition formation game.

One key point in coalition formation is about what strategy to adopt by each D2D pair. In other words, each D2D pair chooses to join one of the coalitions, and then is able to compare and order its potential coalitions based on well-defined preferences. In order to evaluate these preferences, we introduce the concept of preference relation or order in detail \cite{li2014coalitional}, \cite{zhao2015social}.

\begin{definition}
Preference Order

For any D2D pair $i\in \textbf{D}$, the preference relation or order ${\succ}_i$ is defined as a complete, reflexive, and transitive binary relation over the set of all coalitions that D2D pair $i$ can possibly form.
\label{define3}
\end{definition}

Hence, the D2D pairs in our coalitional game have the right to choose to join or leave a coalition according to their preference order, that is to say, the D2D pair tends to join a coalition based on which it prefers to being a member. For any given D2D pair $i\in \textbf{D}$, $F_{c} \ {\succ}_i \ F_{{c}^{\prime}}$ implies that D2D pair $i$ is more willing to be a member of the coalition $F_{c}\subset \textbf{D}$ with $i\in F_{c}$ than $F_{{c}^{\prime}}\subset \textbf{D}$ with $i\in F_{{c}^{\prime}}$, which does not include the case that D2D pair $i$ prefers these two coalitions equally. In different applications, the preferences for D2D pairs can be quantified into different inequalities. In this paper, for any D2D pair $i\in \textbf{D}$ and $i\in F_{c}, F_{{c}^{\prime}}$, we propose the following preference, which is called the utilitarian order \cite{saad2009coalitional}.
\begin{equation}
F_{c} \ {\succ}_{i} \ F_{{c}^{\prime}} \Longleftrightarrow R(F_{c})+R(F_{{c}^{\prime}} \backslash i) > R(F_{c} \backslash i)+ R(F_{{c}^{\prime}}). \label{eq20}
\end{equation}
This definition means D2D pair $i$ prefers being a member of coalition $F_{c}$ than $F_{{c}^{\prime}}$ under the condition that the system sum profit increases.
For forming coalitions based on the above preference order, we define the switch operation as follows.
\begin{definition}
Switch Operation:
Given a partition $F=\{F_{c_{1}},F_{c_{2}},...,F_{c_{C}},F_{c_{C+1}}\}$ of the D2D pairs set $\textbf{D}$, if D2D pair $i\in \textbf{D}$ performs a switch operation from $F_{c}$ to $F_{{c}^{\prime}}$, $F_{c} \neq F_{{c}^{\prime}}$, then the current partition $F$ is modified into a new partition ${F}^{\prime}$ such that ${F}^{\prime}=(F\backslash \{F_{c},F_{{c}^{\prime}}\})\bigcup \{{F_{c}\backslash \{i\},F_{{c}^{\prime}}\bigcup \{i\}}\}$.
\label{define4}
\end{definition}

We initialize the system by any random coalition partition $F=\{F_{c_{1}},F_{c_{2}},...,F_{c_{C}},F_{c_{C+1}}\}$. For any D2D pair $i\in \textbf{D}$, we suppose its current coalition is $F_{c}$, where $F_{c}\subset F$. Then, we uniformly randomly choose another coalition $F_{{c}^{\prime}}$ and suppose the preference relation $F_{{c}^{\prime}} \ {\succ}_{i} \ F_{c}$ is satisfied, where $F_{{c}^{\prime}}\subset F, F_{c}\neq F_{{c}^{\prime}}$, which means a switch operation from $F_{c}$ to $F_{{c}^{\prime}}$ and the current coalition partition will be updated to a new partition ${F}^{\prime}$ as shown in definition \ref{define4}. Actually, the switch operation can be performed if and only if the preference relation defined in (\ref{eq20}) is satisfied. In this mechanism, every D2D pair $i\in \textbf{D}$ can leave its current coalition and join another coalition, given that the new coalition is strictly preferred through the definition in (\ref{eq20}) and the D2D pair can make a greater contribution to the entire system performance in terms of sum rate in the new coalition. In general, our proposed coalition formation game entails finding a coalitional structure that maximizes the total utility rather than the individual payoffs of the players.

\begin{algorithm}[htbp]
\caption{The Coalition Formation Algorithm for the D2D Pairs Resource Allocation}
\label{alg1}
\begin{algorithmic}[1]
\STATE Given any partition $F_{ini}$ of the D2D pairs set $\textbf{D}$;
\STATE Set the current partition as $F_{ini}\longrightarrow F_{cur} , num=0$;
\REPEAT
\STATE Choose one D2D pair $i\in \textbf{D}$ in a pre-determined order, and denote its coalition as $F_{c}\subset F_{cur}$;
\STATE Uniformly randomly search for another possible coalition $F_{{c}^{\prime}}\subset F_{cur},F_{{c}^{\prime}}  \neq  F_{c}$;
\STATE Calculate $R(F_{c})$ and $R(F_{{c}^{\prime}})$;
\IF{The switch operation from $F_{c}$ to $F_{{c}^{\prime}}$ satisfying $F_{{c}^{\prime}} {\succ}_{i} F_{c}$}
\STATE $num=0$;
\STATE D2D pair $i$ leaves its current coalition $F_{c}$, and joins the new coalition $F_{{c}^{\prime}}$;
\STATE Update the current partition set as follows $(F_{cur} \backslash \{F_{{c}^{\prime}},F_{c}\})\bigcup \{F_{c} \backslash \{i\},F_{{c}^{\prime}}\bigcup \{i\}\}\longrightarrow F_{cur}$;
\ELSE
\STATE $num=num+1$;
\ENDIF
\UNTIL The partition converges to the final Nash-stable partition $F_{fin}$.
\end{algorithmic}
\end{algorithm}

The coalition formation game is summarized in Algorithm \ref{alg1}, where the D2D pairs make switch operation in a random order. In the algorithm, we first give any partition $F_{ini}$ of the D2D pairs set $\textbf{D}$. Then, the system will choose one of the D2D pairs in a pre-determined order in step 4. The selected D2D pair saves the coalition $F_{c}$ currently located and then uniformly randomly selects another possible coalition $F_{{c}^{\prime}}$ in step 5. In step 6, the D2D pair obtains the channel information of both coalitions $F_{c}$ and $F_{{c}^{\prime}}$ from BS. Then, it calculates respectively the received sum rate of these two coalitions and makes a decision on whether to perform the switch operation. If the preference relation is satisfied, we update the current coalition partition and reset the number of consecutive unsuccessful switch operations $num$ to zero. Otherwise, we increase the number of consecutive unsuccessful switch operations by 1. When the value of $num$ is equal to multiply the number of D2D pairs by 10 \cite{zhao2015social}, the algorithm stops iterating and performs operations outside the loop. Finally, the system partition will converge to the final Nash-stable partition $F_{fin}$ after a limited number of switching.

\section{Theoretical Analysis}\label{S5}

\subsection{Convergence}\label{S5-1}
In this subsection, the convergence of the proposed coalition formation algorithm is guaranteed as follows \cite{wang2013dynamic}.
\begin{theorem}
Starting from any initial coalitional structure $F_{ini}$, the proposed coalition formation algorithm will always converge
to a final network partition $F_{fin}$, which is consisted by a number of disjoint coalitions, after a sequence of switch operations.
\label{the1}
\end{theorem}
\begin{proof}
Through careful inspection of the preference defined in (\ref{eq20}), we find that each switch operation in Algorithm \ref{alg1} will either yield an unvisited partition through adopting new strategy or switch existing partitions. As a result, part of coalitions may degenerate into the sets of very few D2D pairs, and even be emptied. The system will form at most $C+1$ partitions as there is only $C$ cellular users plus one mmWave band. As the number of partitions for the already given D2D pairs set $\textbf{D}$ is the Bell number \cite{saad2009coalitional}, we draw the conclusion that the sequence of switch operations will always terminate and converge to a final partition $F_{fin}$, which completes the proof that our proposed coalition formation algorithm is convergent.
\end{proof}

\subsection{Stability}\label{S5-2}

In this subsection, we study the stability of the proposed coalition formation algorithm by using the definition from the hedonic games as follows \cite{feng2016reliable}.

\begin{definition}
Nash-stable Structure:
A coalitional partition $F=\{F_{c_{1}},F_{c_{2}},...,F_{c_{C}},F_{c_{C+1}}\}$ is Nash-stable, if $\forall i\in \textbf{D}, i\in F_{c}\subset F, F_{c} \succ_{i} F_{{c}^{\prime}} \bigcup \{i\}$ for all $F_{{c}^{\prime}} \subset F, F_{{c}^{\prime}}\neq F_{c}$.
\label{define5}
\end{definition}
\begin{theorem}
The final partition $F_{fin}$ in our coalition formation algorithm is Nash-stable.
\label{the2}
\end{theorem}
\begin{proof}
The coalition game has the Nash-stable coalitional structure if no D2D pair can make its contribution to the entire system increased by changing its resource sharing strategy.
${F_{c}}^{\ast}$=arg $\max\limits_{F_{c}}R(F), \forall F_{c}\subset F$, and ${F}^{\ast}=\{{F}^{\ast}_{c_{1}},{F}^{\ast}_{c_{2}},...,{F}^{\ast}_{c_{C}},{F}^{\ast}_{c_{C+1}}\}$ is the final Nash-stable coalitional structure.
We prove the stability by contradiction. Assuming that the final formed coalition partition $F_{fin}$ is not Nash-stable. In other words, there exists a D2D pair $i\in \textbf{D}$, and its located coalition currently and randomly selected new coalition are denoted by $F_{c}$ and $F_{{c}^{\prime}}$ respectively.
These two coalitions meet the preference relation $F_{{c}^{\prime}}\bigcup \{i\}\ {\succ}_{i} \ F_{c}$. Consequently, D2D pair will perform the operations leaving its current coalition $F_{c}$ and joining the new coalition $F_{{c}^{\prime}}$, which means that $F_{fin}$ will be updated and it is not the final partition. Thus, we complete the proof that the final partition $F_{fin}$ of our proposed coalition formation algorithm is Nash-stable.
\end{proof}

\subsection{Optimality}\label{S5-3}
\begin{theorem}
The solution obtained by our proposed algorithm corresponds to an optimal system performance.
\label{the3}
\end{theorem}
\begin{proof}
The total utility achieved by our proposed coalition formation algorithm is convergent with a sufficiently large number of iterations. In Algorithm \ref{alg1}, we set the termination condition to be that the number of consecutive unsuccessful switch operations $num$ is equal to the product of the number of D2D pairs and 10. On the other hand, our scheme only involves one-step switching and it has the limitation of allowing multiple D2D pairs to perform switch operations simultaneously. Thus, the solution obtained by Algorithm \ref{alg1} is near-optimal compared with the solution obtained by the exhaustive search method. From the Fig. \ref{fig4}(a) and Fig. \ref{fig4}(b) in Section~\ref{S6}, the gap between our scheme and the optimal solution is quite small and the performance of our proposed algorithm is guaranteed. Besides, the in-depth analysis of the performance bound will be carried out in the future work.
\end{proof}

\subsection{Complexity}\label{S5-4}
\begin{theorem}
Given the total number of iterations $N$, the computational complexity of Algorithm \ref{alg1} can be approximated as $O (N)$.
\label{the4}
\end{theorem}
\begin{proof}
In each iteration of Algorithm \ref{alg1}, the selected D2D pair calculates the total utility of currently located coalition and another possible coalition, respectively. Then, it makes a decision on whether to perform a switch operation. Thus, there is at most 1 switch operation to be considered in each iteration, and the complexity lies in the number of iterations. From the Fig. \ref{fig13}, we can see the computational complexity of Algorithm \ref{alg1} is extremely low.
\end{proof}
\begin{table}[htbp]
\begin{center}
\caption{SIMULATION PARAMETERS}
\begin{tabular}{ccc}
\hline
Parameter & Symbol  & Value \\
\hline
mmWave bandwidth & $W_{m}$ & 2160 MHz  \\
Cellular carrier bandwidth & $W_{c}$ & 15 KHz \\
mmWave noise spectral density &  $N_{0m}$  & -134 dBm/MHz \\
Cellular noise spectral density & $N_{0c}$ & -174 dBm/Hz \\
mmWave transmission power & $P_{m}$ & 20 dBm \\
Cellular transmission power&  $P_{c}$ & 23 dBm \\
Path loss exponent  &  $n$   &  2 \\
MUI factor &  $\rho$  & 1 \\
Half-power beamwidth  &   $\theta_{-3dB}$ &    $30^{\circ}$ \\
Blockage parameter  &   $\beta$   & 0.01 \\
Antenna gains of device &   $G_{0}$  &    0.5 dBi \\
Antenna gains of BS &  $G_{b}$    &  14 dBi \\
Maximum distance of D2D  &   $r$   & 10$\sqrt{2}$ m\\
\hline
\end{tabular}
\label{table1}
\end{center}
\end{table}

\section{Performance Evaluation}\label{S6}

In this section, we evaluate the performance of our proposed coalition game under various system parameters. Specially, we compare our scheme with other four schemes in terms of system sum rate. Besides, we give the necessary analysis for the obtained simulation results.

\subsection{Simulation Setup}\label{S6-1}

In the simulation, we consider a single cell scenario, where D2D pairs and cellular users are uniform randomly distributed in a square area of $500m\times 500m$ with the base station in the center. For a fixed number of cellular users and D2D pairs, we repeat the simulation by 20 times and then average the results of positions in order to obtain a more reliable location layout. Not only the path-loss model is considered for cellular and D2D links, but also the shadow fading. Besides, we set the path-loss exponent in free space propagation model to be 2. On the one hand, when two D2D users are physically in close proximity, the D2D communication channel is established. In our simulation, we provide an upper bound on the distance between two D2D users. On the other hand, the widely used realistic directional antenna model is adopted in mmWave D2D network, which is a main lobe of Gaussian form in linear scale and constant level of side lobes \cite{niu2017energy}. Based on this model, the gain of a directional antenna in units of decibel (dB), denoted by $G(\theta)$, can be expressed as

\begin{equation}
G({\theta})=\left\{
\begin{array}{rcl}
G_{0}-3.01\cdot\left({\frac{2\theta}{\theta_{-3dB}}}\right)^2,&& {0^{\circ}\leq {\theta} \leq {{\theta}_{ml}}/2};\\
G_{sl},&& {{{\theta}_{ml}}/2 \leq {\theta} \leq 180^{\circ}};\\
\end{array} \right. \label{eq21}
\end{equation}
where $\theta$ denotes an arbitrary angle within the range $[0^{\circ},180^{\circ}]$, ${\theta}_{-3dB}$ denotes the angle of the half-power beam width, and ${\theta}_{ml}$ denotes the main lobe width in units of degrees. The relationship between ${\theta}_{ml}$ and ${\theta}_{-3dB}$ is ${\theta}_{ml}=2.6\cdot {\theta}_{-3dB}$. $G_{0}$ is the maximum antenna gain, and can be expressed as
\begin{equation}
G_{0}=10\log\left(\frac{1.6162}{\sin(\theta_{-3dB}/2)}\right)^2.
\label{eq22}
\end{equation}
$G_{sl}$ denotes the side lobe gain, which can be obtained by
\begin{equation}
G_{sl}=-0.4111\cdot \ln(\theta_{-3dB})-10.579.
\label{eq23}
\end{equation}

The simulation parameters are summarized in Table \ref{table1} \cite{niu2017energy}. In order to illustrate how cellular and D2D users are distributed and how to share resources, we plot the positions of the base station, cellular users and D2D pairs together in an instance by randomly generating a network consisting of 5 cellular users and 30 D2D pairs in Fig. \ref{fig3}. Besides, we show a snapshot of a final coalition structure resulting from our coalitional formation algorithm. In the figure, the base station, cellular users and D2D pairs are represented by pentacle, triangle and circle respectively. Five cellular users and thirty D2D pairs form six coalitions, and they are marked by different colors of red, green, cyan, dark, yellow and magenta, respectively.

In order to show the advantage of our proposed coalition game in improving system performance in terms of system sum rate $R(F)$, which includes the communication rates of all cellular users and D2D pairs, we compare our scheme, labeled as \textbf{Coalition Game} (CG), with four other schemes :

$a)$ \textbf{Full MmWave Communication} (FMC), where all the D2D pairs are interconnected via direct D2D communications in mmWave band, and each cellular user occupies one of the cellular carrier channels without spectrum sharing.

$b)$ \textbf{Random Communication} (RC), where the system allocates the communication resources to the D2D pairs in a uniform randomly manner. In other words, for any D2D pair, the system randomly selects a cellular user's spectrum resource or the resource in mmWave band.

$c)$ \textbf{Cellular Coalition Game} (CCG), which utilizes coalition game to cope with the problem of the resource allocation among cellular bands for multiple D2D pairs in cellular network. In order to maximize the system total utility, the algorithm performs switch operations based on well-defined preference order with a limited number of iterations.

$d)$ \textbf{Full Cellular Communication} (FCC), which uniform randomly allocates cellular users' uplink spectrum resources to the D2D pairs. Generally speaking, this kind of method is similar to RC, and the difference is that this scheme does not involve mmWave. Since the transmission rate of cellular communication is much smaller than that of mmWave communication, this kind of method represents the worst case of the system performance in terms of sum rate compared with above methods.
\begin{figure}[htbp]
\begin{center}
\includegraphics*[width=0.9\columnwidth,height=2.5in]{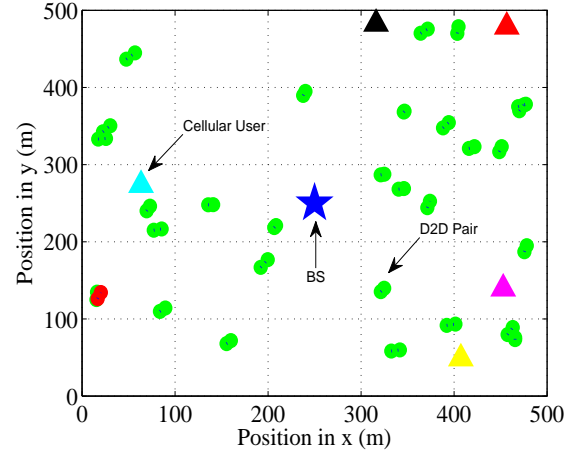}
\end{center}
\caption{A snapshot of a final coalition structure resulting from CG for a network of 5 cellular users and 30 D2D pairs.}
\label{fig3}
\end{figure}

\subsection{Compared With the Optimal Solution}\label{S6-3}

\begin{figure*}[htbp]
\begin{minipage}[t]{0.5\linewidth}
\centering
\includegraphics[width=0.9\columnwidth,height=2.5in]{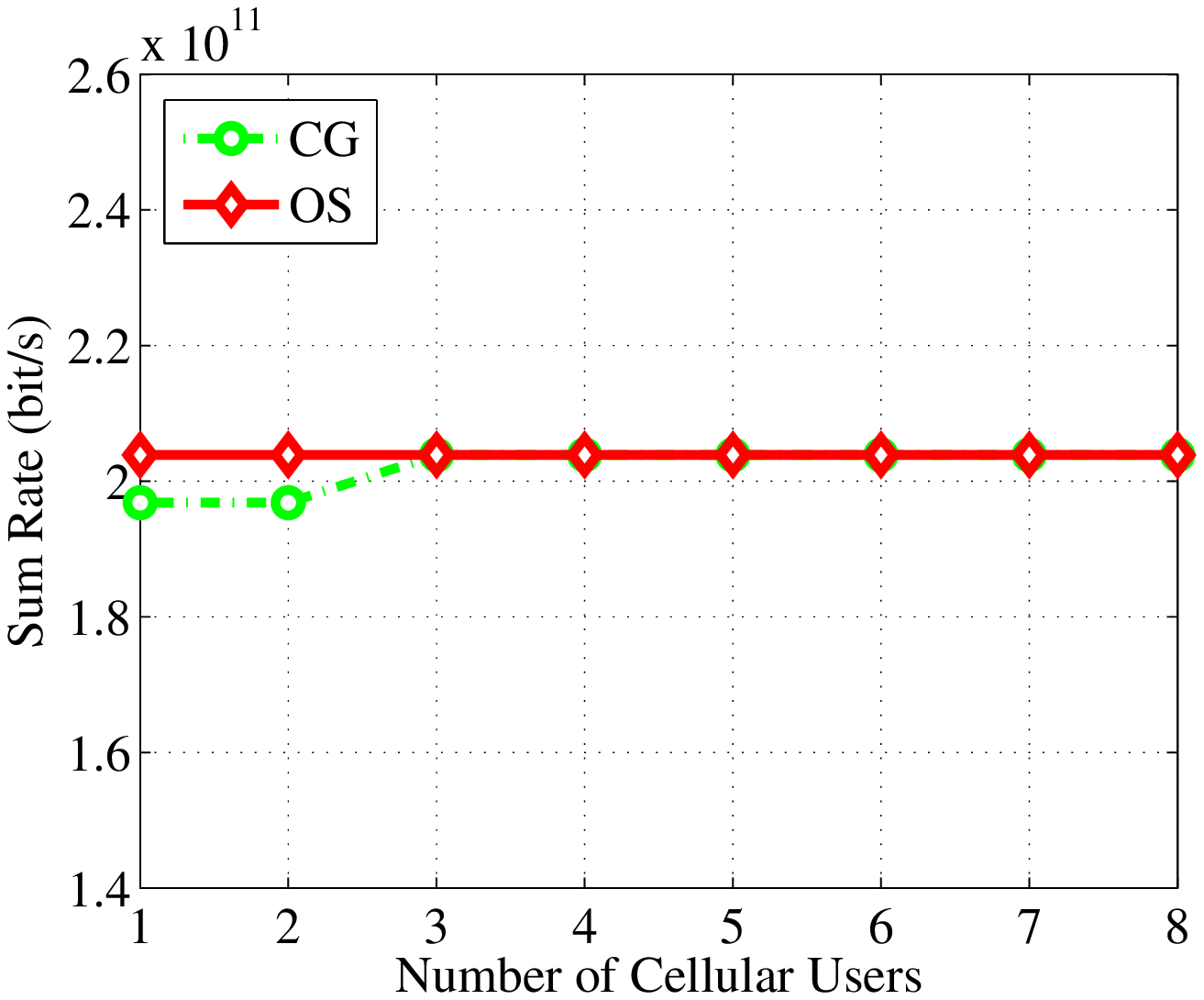}
\centerline{\small (a)}
\end{minipage}%
\begin{minipage}[t]{0.5\linewidth}
\centering
\includegraphics[width=0.9\columnwidth,height=2.5in]{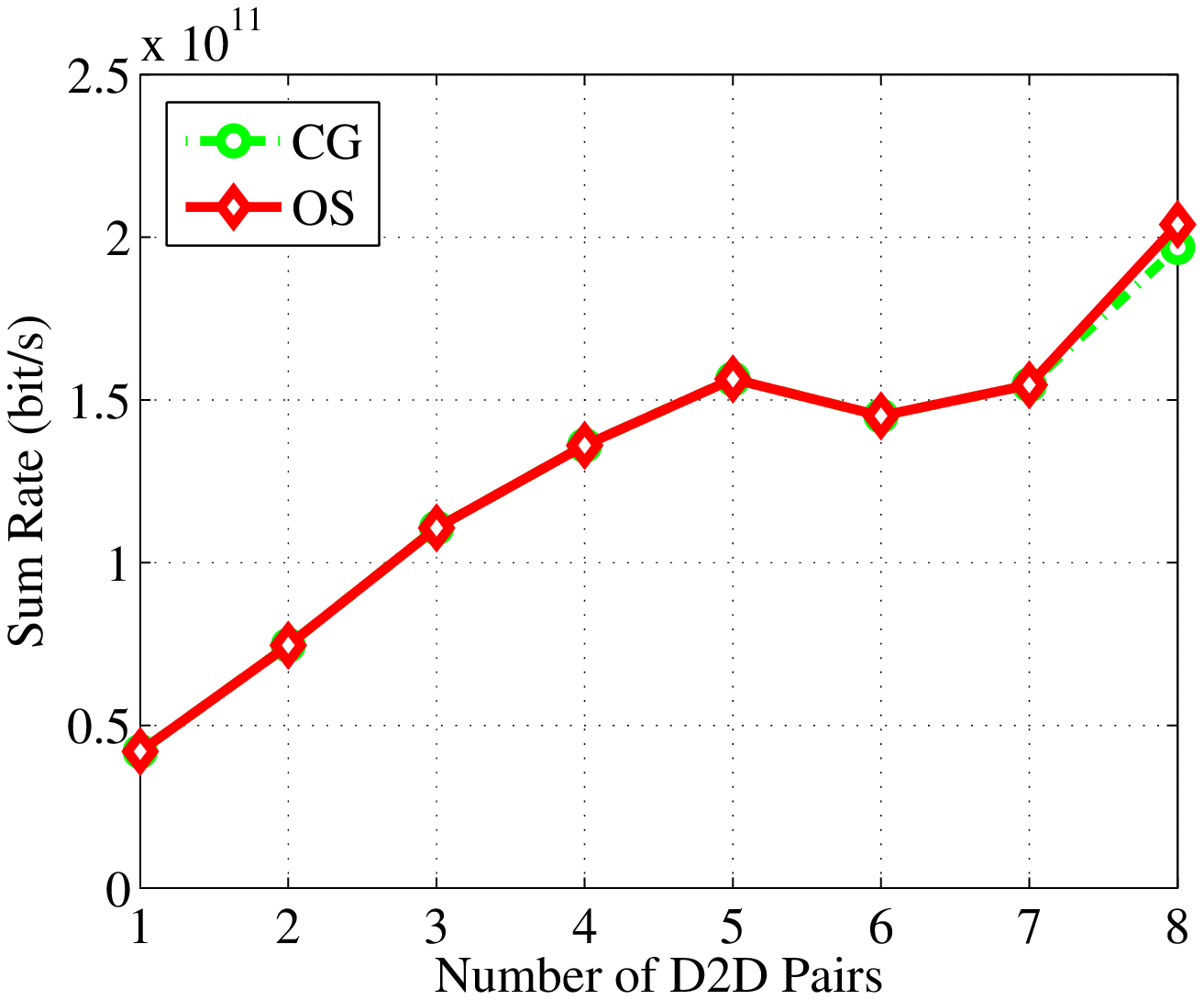}
\centerline{\small (b)}
\end{minipage}
\caption{System sum rate comparison of CG and OS with (a) different number of cellular users and (b) different number of D2D pairs.}
\vspace*{-3mm}
\label{fig4}
\end{figure*}

In this subsection, we compare the performance of CG with the optimal solution, labeled as OS, which is obtained by the traditional exhaustive search method. In view of the highly complexity of this method, we set the number of D2D pairs to be 10 and vary the number of cellular users to be 1 to 8 to obtain the simulation results shown in Fig. \ref{fig4}(a), while set the number of cellular users to be 1 and vary the number of D2D pairs to be 1 to 8 to obtain the simulation results shown in Fig. \ref{fig4}(b). From these two figures, we can see the system sum rate achieved by CG, shown by the dot and dash curve, has an excellent approximation to that achieved by OS, shown by solid line curve. In order to further demonstrate our proposed scheme CG converges close to the OS, we analyze the simulation results in detail and calculate the average deviation between the results obtained by CG and OS, which is expressed as follows.

\begin{equation}
Average \ Deviation = \frac{1}{8}\sum\limits_{n=1}^{8}{\frac{R_{OS}(n)-R_{CG}(n)}{R_{OS}(n)}}, \label{eq24}
\end{equation}
where $R_{OS}(n)$ and $R_{CG}(n)$ denote the system sum rate obtained by OS and CG, respectively, with the number of cellular users or D2D pairs $n$. As a result, the average deviation between the CG and OS is about $0.9\%$ in Fig. \ref{fig4}(a), while the average deviation is about $0.4\%$ in Fig. \ref{fig4}(b). Thus, we complete the demonstration that our proposed coalition game can achieve the system sum rate which is close to the optimal solution of the resource allocation problem.
\subsection{System Sum Rate}\label{S6-3}

In this subsection, we evaluate the performance of our proposed coalition game based resource allocation scheme under various system parameters, and then demonstrate the advantage of this algorithm compared with four other schemes.

\begin{figure}[htbp]
\begin{center}
\includegraphics*[width=0.9\columnwidth,height=2.5in]{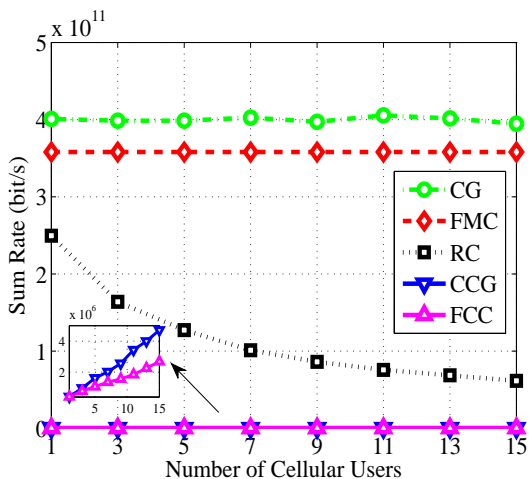}
\end{center}
\caption{System sum rate comparison of five resource allocation algorithms with different number of cellular users.}
\label{fig5}
\end{figure}
In Fig. \ref{fig5}, we set the number of D2D pairs to be 30 and the other parameter settings are shown in Table \ref{table1}. Then, we plot the system sum rate comparison of five schemes varying the number of cellular users from 1 to 15. From the figure, we can observe that the system sum rate of CG has almost no changes as the number
of cellular users increases. It is because that the mmWave communication rate is much greater than that of cellular communication, which results in the increase in the number of D2D pairs using the spectrum resource in mmWave band in order to maximize the system sum rate. In other words, the utility contributed by cellular users and D2D pairs in cellular band accounts for a very small proportion of total system utility. At the same time, the randomness of the CG leads to slight fluctuations in the curve. Comparing these five schemes, the system sum rate received by adopting CG is much larger than other schemes. When the number of cellular users is equal to 15, the sum rate of CG is larger than that of FMC and RC about $10\%$ and $543\%$, respectively. In addition, with the number of cellular users increased, more D2D pairs will uniform randomly choose to share the spectrum resources with cellular users and the number of D2D pairs using the resource in mmWave band is decreased, which explains the change in the RC curve. CCG and FCC increase as the number of cellular users increases. The reason is that the bandwidth resource for the D2D transmission increases. Meanwhile, the cellular users still make a contribution to the system sum rate. For CG, FMC and RC, involving mmWave D2D communications can offload cellular traffic and improve the system performance at the same time.
\begin{figure}[htbp]
\begin{center}
\includegraphics*[width=0.9\columnwidth,height=2.5in]{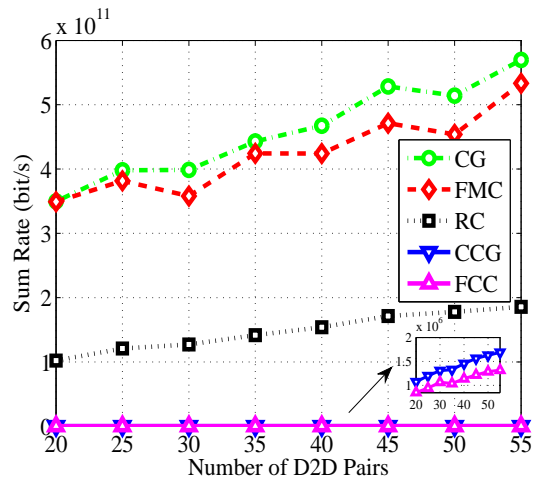}
\end{center} \caption{System sum rate comparison of five resource allocation algorithms with different number of D2D pairs.}
\label{fig6}
\end{figure}

In Fig. \ref{fig6}, we set the number of cellular users to be 5 and vary the number of D2D pairs to be 20 to 55. From the figure, we can see the proposed CG algorithm performs much better than other schemes. When the number of D2D pairs is equal to 55, the sum rate of CG is larger than that of FMC and RC about $7\%$ and $207\%$, respectively. Fig. \ref{fig6} indicates the system sum rate of five schemes increases as the number of D2D pairs increases.
At the same time, different number of D2D pairs makes the change of positions in each simulation, which leads to individual drop points in CG and FMC. With more D2D pairs included in the network, the spectrum utilization can be improved, while the interference caused by spectrum sharing also increases, which constraints the system performance. Besides, the FCC still gets the worst performance, while the FMC, RC and CCG achieve the middle performance.

\begin{figure}[htbp]
\begin{center}
\includegraphics*[width=0.9\columnwidth,height=2.5in]{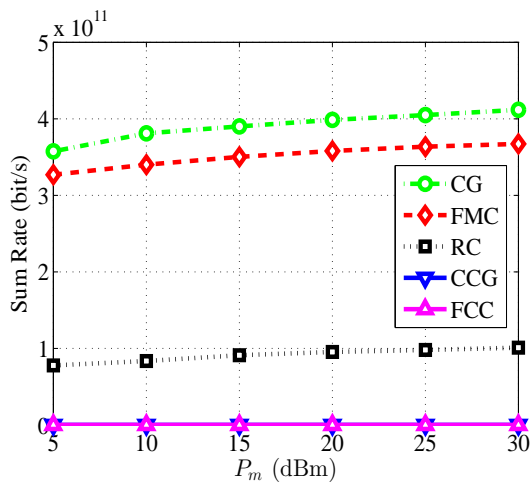}
\end{center}
\caption{System sum rate comparison of five resource allocation algorithms with different $P_{m}$.}
\label{fig7}
\end{figure}

\begin{figure}[htbp]
\begin{center}
\includegraphics*[width=0.9\columnwidth,height=2.5in]{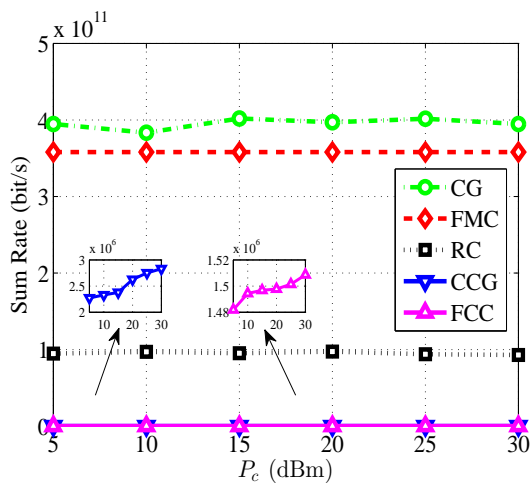}
\end{center}
\caption{System sum rate comparison of five resource allocation algorithms with different $P_{c}$.}
\label{fig8}
\end{figure}

We set the number of cellular users and D2D pairs to be 8 and 30. Fig. \ref{fig7} indicates the system sum rate increases with the mmWave transmission power $P_{m}$ varied from 5 to 30 dBm in CG, FMC and RC. These three curves grow slowly as the result that the corresponding interference power increases and the improvement in sum rate would be less with the $P_{m}$ increased. Compared the behaviors of different schemes, we observe that the CG obtains the highest system sum rate. FCC obtains the lowest system sum rate, while FMC, RC and CCG perform medially. When the mmWave transmission power $P_{m}$ is equal to 30 dBm, the sum rate of CG is larger than that of FMC and RC about $12\%$ and $307\%$, respectively.

Similarly, Fig. \ref{fig8} indicates the system sum rate of CCG and FCC increases with the cellular transmission power $P_{c}$ varied from 5 to 30 dBm, while the effect of $P_{c}$ on CG, FMC and RC is not significant.
When the cellular transmission power $P_{c}$ is equal to 30 dBm, the sum rate of CG is larger than that of FMC and RC about $10\%$ and $325\%$, respectively.

\begin{figure}[htbp]
\begin{center}
\includegraphics*[width=0.9\columnwidth,height=2.5in]{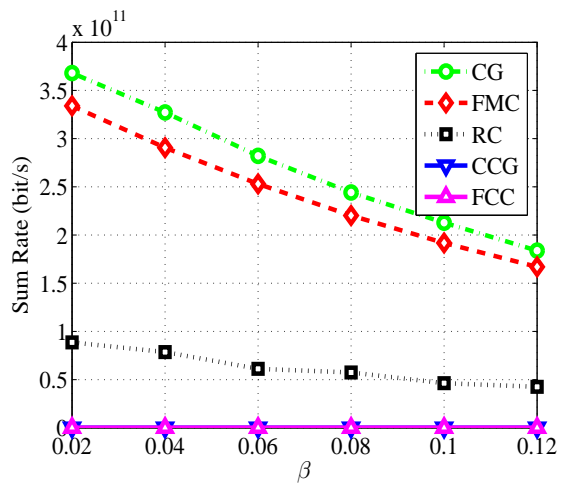}
\end{center}
\caption{System sum rate comparison of five resource allocation algorithms with different $\beta$.}
\label{fig9}
\end{figure}

\begin{figure}[htbp]
\begin{center}
\includegraphics*[width=0.9\columnwidth,height=2.5in]{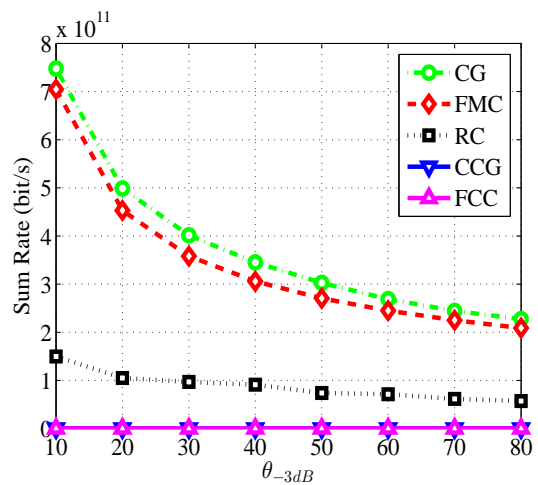}
\end{center}
\caption{System sum rate comparison of five resource allocation algorithms with different $\theta_{-3dB}$.}
\label{fig10}
\end{figure}

In Fig. \ref{fig9}, we set the number of cellular users and D2D pairs to be 8 and 30, respectively, and then vary the $\beta$ from 0.02 to 0.12. In terms of the impact of the blockage parameter that captures the density and size of obstacles, we observe that our proposed scheme again has the best performance. When the blockage parameter $\beta$ is equal to 0.12, the sum rate of CG is larger than that of FMC and RC about $10\%$ and $332\%$, respectively. The greater $\beta$ means obstacles with higher density and larger size, which results in higher blockage probability. In other words, the rate of the mmWave communication channel shared by D2D pairs decreases due to unreliable direct D2D connectivity with $\beta$ increased, which explains the changes of CG, FMC and RC. Besides, the system sum rate of CCG and FCC keeps at a low level and they are not affected by changing $\beta$.

In Fig. \ref{fig10}, we set the number of cellular users and D2D pairs to be 8 and 30, respectively, and then plot the system sum rate comparison of five resource allocation algorithms varying ${\theta}_{-3dB}$ from 10 to 80. The parameter of ${\theta}_{-3dB}$ denotes the angle of the half-power beamwidth adopting the widely used realistic directional antenna model in mmWave D2D network. As the ${\theta}_{-3dB}$ increases, the system sum rate of CG, FMC and RC decreases. This is because the antenna with larger beamwidth covers the wider area, which causes greater interference toward other D2D pairs in mmWave band, and furthermore results in the changes in Fig. \ref{fig10}. From the figure, we can see the CG performs better than other schemes. When the ${\theta}_{-3dB}$ is equal to 80, the sum rate of CG is larger than that of FMC and RC about $9\%$ and $298\%$, respectively.

\begin{figure}[htbp]
\begin{center}
\includegraphics*[width=0.9\columnwidth,height=2.5in]{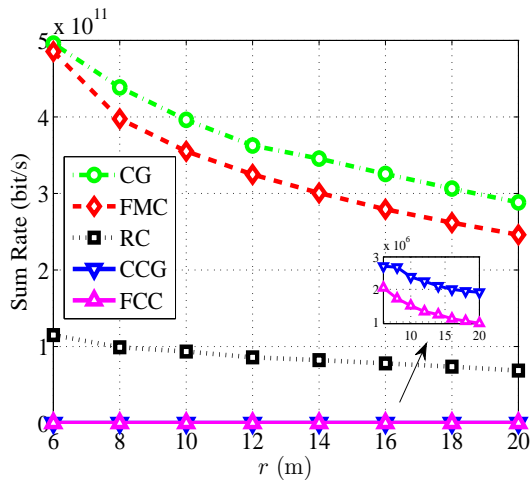}
\end{center}
\caption{System sum rate comparison of five resource allocation algorithms with different $r$.}
\label{fig11}
\end{figure}

\begin{figure}[htbp]
\begin{center}
\includegraphics*[width=0.9\columnwidth,height=2.5in]{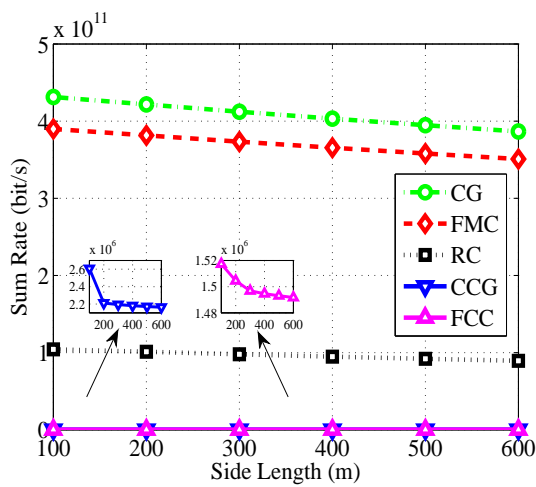}
\end{center}
\caption{System sum rate comparison of five resource allocation algorithms with different side length.}
\label{fig12}
\end{figure}

In Fig. \ref{fig11}, we set the number of cellular users and D2D pairs to be 8 and 30, respectively, and then plot the system sum rate comparison of five resource allocation algorithms varying the maximum distances of both abscissa and ordinate between D2D users from 6 to 20. On the one hand, whether the cellular communication mode, or mmWave communication mode, the increase of $r$ will make the path loss more serious, and thus make the system performance of CG, FMC, RC, CCG and FCC decreased. On the other hand, as the maximum distance of D2D is enlarged, the blockage probability of mmWave communication link increases and the reliability of the link decreases, which furthermore hampers the performance of CG, FMC and RC.

\begin{figure}[htbp]
\begin{center}
\includegraphics*[width=0.9\columnwidth,height=2.5in]{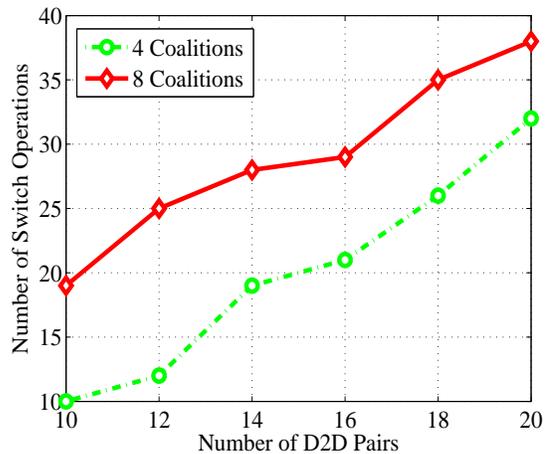}
\end{center}
\caption{System convergence rate in terms of the number of switch operations with different number of D2D pairs.}
\label{fig13}
\end{figure}

In Fig. \ref{fig12}, we set the number of cellular users and D2D pairs to be 8 and 30, respectively, and then plot the system sum rate varying the side length of the square from 100 to 600. In order to obtain simulation results under practical scenarios, we modify the maximum distance of D2D as 2$\sqrt{2}$, 4$\sqrt{2}$, 6$\sqrt{2}$, 8$\sqrt{2}$, 10$\sqrt{2}$ and 12$\sqrt{2}$, respectively. As the side length of the square is enlarged, or equivalently the user distribution density is decreased, the path loss of all links and the blockage probability of mmWave communication links are increased, which directly reduces the system sum rate of all schemes.

\subsection{Convergence Rate}\label{S6-4}

In order to show the convergence rate of our proposed algorithm, we set the number of cellular users to be 3 and 7, or equivalently the number of coalitions to be 4 and 8, and vary the number of D2D pairs to be 10 to 20. In Fig. \ref{fig13}, we show the number of switch operations of CG converging to the final partition. From the figure, we observe that the number of switch operations increases with the number of coalitions or D2D pairs increased. In the cases of 3 and 7 cellular users, the average number of switch operations is from 10 to 32 and 19 to 38, respectively. For the exhaustive search method, each D2D pair can choose to join one of the 8 coalitions when there exists 7 cellular users. The exhaustive search method needs $8^{N}$ iterations to find the optimal solution as the number of D2D pairs is set to be $N$. Therefore, our proposed coalition game algorithm allows D2D pairs and cellular users to form the final Nash-stable partition with extremely fast convergence rate and decreases the computation complexity significantly.
\section{Conclusion}\label{S7}

In this paper, we investigate the problem of maximizing the system sum rate via resource allocation for D2D communications underlaying HCN combining mmWave and the traditional cellular band. After formulating the problem of the uplink resource allocation among mmWave and the cellular band for multiple D2D pairs and cellular users into a non-linear integer programming problem, we propose a coalition game based approach to obtain the near-optimal solution. Through extensive simulations under various system parameters, we demonstrate the superior performance of our proposed coalition game in terms of sum rate compared with four other practical schemes.

\end{document}